\documentclass[aps, pra, twocolumn, superscriptaddress, 10pt]{revtex4-2}
\usepackage{amsthm}
\usepackage{amsmath,bm}
\usepackage{amssymb}
\usepackage{amsfonts}
\usepackage{graphicx}
\usepackage{txfonts}
\usepackage{xcolor}
\usepackage{braket}
\usepackage{bbm}
\usepackage{subcaption}
\usepackage{changes}
\usepackage{ragged2e}
\usepackage[colorlinks=true,linkcolor=blue,citecolor=blue,urlcolor=blue]{hyperref}
\usepackage[capitalise]{cleveref}
\usepackage{booktabs}
\usepackage{diagbox}

\newcommand{\abs}[1]{\ensuremath{\left|{#1}\right|}}
\newcommand{\avg}[1]{\langle#1\rangle}		

\newcommand{\vect}[1]{{\bm{#1}}}			
 
\newcommand{\dd}{\text{d}} 
\newcommand{\id}{\mathbb{I}}

\newcommand{\ignore}[1]{} 

\newcounter{SaveEqnCntr}

\newcommand{\be}{\begin{equation}}
\newcommand{\ee}{\end{equation}}
\newcommand{\ba}{\begin{eqnarray}}
\newcommand{\ea}{\end{eqnarray}}
\newcommand{\proj}[1]{\ket{#1}\bra{#1}}

\def\>{\rangle}
\def\<{\langle}

\newcommand{\Tr}{\text{Tr}}

\newtheorem{theorem}{Theorem}
\newtheorem{corollary}{Corollary}

\newtheorem{proposition}{Proposition}

\newtheorem{lemma}{Lemma}

\newcommand{\Nc}{\mathcal{N}_{\mathrm{c}}}
\newcommand{\No}{\mathcal{N}_{\mathrm{o}}}
\newcommand{\Sv}{\mathcal{S}}
\newcommand{\Me}{\mathcal{M}}
\newcommand{\GHZ}{\mathrm{GHZ}}
\usepackage{dsfont}
\renewcommand{\id}{\mathds{1}}
\newcommand{\norm}[1]{\lVert#1\rVert}

\begin{document}
\title{Network nonlocality with tripartite entangled sources}

\author{Sudip Chakrabarty}
\email{sudip27042000@gmail.com}
\affiliation{S. N. Bose National Centre for Basic Sciences, Block JD, Sector III, Salt Lake, Kolkata 700 106, India}


\author{Amit Kundu}
\email{amit8967@gmail.com}
\affiliation{S. N. Bose National Centre for Basic Sciences, Block JD, Sector III, Salt Lake, Kolkata 700 106, India}
\affiliation{Optics and Quantum Information Group, The Institute of Mathematical Sciences, CIT Campus, Taramani, Chennai 600113, India}


\begin{abstract}
Quantum networks are built from independent sources, and their nonlocality has mostly been studied with sources that emit bipartite entangled states. A source of three entangled qubits is richer, since its nonlocality can be tested by the Mermin or the more demanding Svetlichny inequality. We ask how network nonlocality relates to these two single-source notions, and find that the answer depends on the wiring. When two central nodes without inputs share every source, Bell-state measurements turn every Svetlichny violation into a network violation, and even all entangled generalized Greenberger-Horne-Zeilinger states violate the network inequality. When a single central node is shared, the relation reverses, and a network violation requires a Mermin violation of the source. With more sources, the maximal violation depends on whether their number is even or odd, and we give tight bounds for both cases. These results show that the way multipartite sources are wired decides which kind of nonlocality a network inherits.
\end{abstract}

\maketitle
\section{Introduction}
\label{sec:intro}

Violation of Bell's inequality implies that the correlations of a single entangled state cannot be explained by a local hidden variable shared among the observers~\cite{Bell1964,CHSH1969,Brunner2014}, giving rise to the concept of nonlocality. Beyond its foundational significance, Bell nonlocality also enables practical tasks such as device-independent protocols~\cite{mckague, mckague12, supicrev}, quantum cryptography~\cite{acin07, Farkas2024, QCrev, Gras2021}, communication complexity~\cite{Rev1}, and certification of randomness~\cite{Piro2010, Acin2012}, and thus lies at the heart of quantum information science~\cite{Brunner2014}. A quantum network changes the central structure of the scenario. Here several independent sources distribute particles to different nodes and edges of the network. In the nodes the experimentalist can do joint measurement on particles produced by different independent sources. The natural local model assigns one hidden variable to each source and keeps these variables independent~\cite{Branciard2010,Branciard2012,Fritz2012}. Correlations that escape every such network local model are called network nonlocal. Since the sources are independent, the local correlations form a nonconvex set, the inequalities that bound them are nonlinear. A node can reveal nonlocality even if it always performs the same joint measurement~\cite{Branciard2012,Fritz2012,Renou2019}. The bilocal scenario of entanglement swapping~\cite{Zukowski1993,Branciard2010}, its extensions to star and other topologies~\cite{Tavakoli2014,Mukherjee2015,Rosset2016,Chaves2016,Tavakoli2017,Andreoli2017, kundu_2020, doolittle2023, chakrabarty_CV_network, Munshi2021} and the first experiments~\cite{Saunders2017} have made network nonlocality a central part of the theory of the quantum internet~\cite{Wehner2018,Tavakoli2022}, which is one of the central goals of recent quantum technologies.

However, almost all of the existing works on network nonlocality assume that every source emits a pair of particles. For two qubits there is a single notion of nonlocality, and the bilocality inequality is tied to it in a simple way. Two identical copies of a two-qubit state violate the bilocality inequality exactly when one copy violates the Clauser-Horne-Shimony-Holt (CHSH) inequality~\cite{Gisin2017,Andreoli2017}. A source that emits three particles is different. Three-qubit entanglement comes in the inequivalent GHZ and $W$ classes~\cite{GHSZ1990,Dur2000}, and nonlocality itself splits into several notions. The Mermin inequality~\cite{Mermin1990} excludes fully local models. The Svetlichny inequality~\cite{Svetlichny1987,Seevinck2002,Collins2002} also excludes models in which two of the three parties share arbitrary resources, so its violation certifies genuine tripartite nonlocality~\cite{Gallego2012,Bancal2013}. A state may violate the first and satisfy the second.

Network nonlocality with multipartite sources have been studied in a very limited number of works before. Trilocal inequalities for three tripartite sources were derived in Ref.~\cite{Mukherjee2017}, where it was also noticed that genuine tripartite entanglement is not needed for a violation. General nonlinear inequalities that cover GHZ sources were given in Ref.~\cite{Luo2018}. Closest to the present work, Ref.~\cite{Mukherjee2020} considered three tripartite sources whose particles are measured by two intermediate parties with a fixed GHZ-basis measurement. 
These works show that such networks can be nonlocal and that their violation can signal entanglement. See \cref{tab:literature} for a more detailed comparison with these related works. However, nonlocality in multipartite systems have different notions, making it distinct from the ususal Bell-CHSH scenario. All the mentioned related works do not tell us how the network nonlocality compares with the nonlocality of a single tripartite source.

\begin{table*}[t]
\caption{\justifying\small Network nonlocality with central nodes, and what the present work adds. ``Central node'' means a node that measures particles from several sources jointly.}
\label{tab:literature}
\begin{ruledtabular}
\begin{tabular}{llll}
Reference & Sources & Inputs of central nodes & Main result\\
\hline
\parbox[t]{0.15\textwidth}{\raggedright Mukherjee \textit{et al.}~\cite{Mukherjee2017}} & \parbox[t]{0.16\textwidth}{\raggedright three tripartite} & \parbox[t]{0.19\textwidth}{\raggedright two settings} & \parbox[t]{0.38\textwidth}{\raggedright trilocal inequalities, violation without genuine tripartite entanglement}\\[2pt]
\parbox[t]{0.15\textwidth}{\raggedright Luo~\cite{Luo2018}} & \parbox[t]{0.16\textwidth}{\raggedright general, including GHZ} & \parbox[t]{0.19\textwidth}{\raggedright general} & \parbox[t]{0.38\textwidth}{\raggedright efficiently computable nonlinear inequalities for general networks}\\[2pt]
\parbox[t]{0.15\textwidth}{\raggedright Mukherjee \textit{et al.}~\cite{Mukherjee2020}} & \parbox[t]{0.16\textwidth}{\raggedright bipartite chain; three tripartite} & \parbox[t]{0.19\textwidth}{\raggedright fixed Bell-state or GHZ-basis measurements} & \parbox[t]{0.38\textwidth}{\raggedright trilocal inequalities for the closed network with three sources, violation by GHZ, $W$ and some biseparable states, entanglement detection}\\[2pt]
\parbox[t]{0.15\textwidth}{\raggedright This work} & \parbox[t]{0.16\textwidth}{\raggedright tripartite, closed and open wirings} & \parbox[t]{0.19\textwidth}{\raggedright none} & \parbox[t]{0.38\textwidth}{\raggedright relation to the Mermin and Svetlichny values of one source, joint heralding of swapped entanglement, tight bounds for two-correlator inequalities}\\[2pt]
\end{tabular}
\end{ruledtabular}
\end{table*}

This leads to the question we address in this work. \emph{When independent sources emit three-qubit states, how is a violation of the network inequality related to the Mermin and Svetlichny nonlocality of a single source, and does this relation depend on how the sources are wired?} The answer is not obvious. A network test uses several copies of the source and joint measurements, so it might be stronger or weaker than either single-copy test.
We study the two simplest ways of wiring tripartite sources, shown in \cref{fig:networks}. In the closed network, two central nodes, Alice and Charlie, each receive one particle from every source and always perform the same measurement. The third particle of each source goes to an end node with two measurement settings. In the open network, a single central node receives one particle from every source, and the other two particles of each source go to two end nodes.

Our main result is that the two wirings are bounded by the two single-copy inequalities from opposite sides. In the closed network with Bell-state measurements, the two-copy value of the network inequality is never smaller than the normalized Svetlichny value of one copy (\cref{thm:svet}). Svetlichny nonlocality of the source is therefore sufficient for network nonlocality, and it is not necessary. All entangled generalized GHZ states, noisy $W$ states that satisfy the Svetlichny inequality and even biseparable states violate the network inequality. In the open network the relation is reversed. With a Bell-state measurement at the central node, a violation implies that one copy of the source violates the Mermin inequality (\cref{thm:mermin}).

Two further results follow from this analysis. First, for GHZ sources the entanglement that the closed network swaps to its end nodes is heralded only by the joint outcome of the two central nodes. Given the outcome of one central node alone, the end nodes share a separable state (\cref{prop:swap}). Second, we ask what happens with more sources. Because the central nodes have no input, the two correlators of the network inequality must be read from one fixed measurement. For an odd number of sources they involve anticommuting observables, and the maximal violation of the inequalities drops below its value for an even number of sources. We derive tight bounds for both cases (\cref{thm:parity}).


The paper is organized as follows. \Cref{sec:prelim} reviews the Mermin and Svetlichny inequalities and the network model, and fixes the measurement classes used throughout. \Cref{sec:closed} treats the closed network and its relation to the Svetlichny inequality. \Cref{sec:swap} looks at the state that the closed network prepares at its end nodes. \Cref{sec:open} treats the open network and its relation to the Mermin inequality. \Cref{sec:nlocal} turns to many sources and the parity effect. \Cref{sec:conclusions} discusses what the results mean and concludes. Proofs are given in Appendix~\cref{app:ineq,app:value,app:svet,app:mermin,app:parity,app:swap}, and the numerical methods in Appendix~\ref{app:numerics}.

\begin{figure*}[t]
\centering
\includegraphics[width=0.98\textwidth]{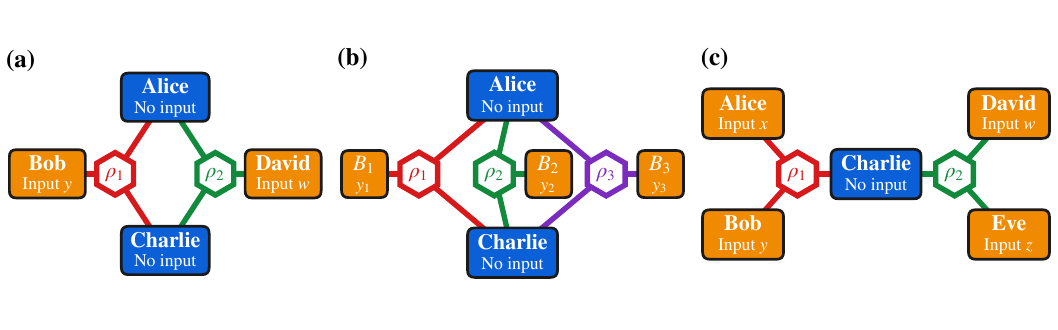}
\caption{\justifying\small Networks with tripartite sources. Hexagons (red, green and purple) are independent sources, each emitting three particles. Blue boxes are central nodes, which perform one fixed measurement. Orange boxes are end nodes, which choose between two dichotomic measurements. \emph{(a) Closed network with two sources.} The central nodes Alice and Charlie hold one particle from each source, and the end nodes Bob and David hold the third particles. \emph{(b) Closed network with three sources and end nodes $B_1$, $B_2$, $B_3$. (c) Open network.} The central node Charlie holds one particle from each source, and the other two particles of every source go to two end nodes.}
\label{fig:networks}
\end{figure*}

\section{Preliminaries}
\label{sec:prelim}

This section collects the three ingredients that the rest of the paper is built on. \Cref{sec:prelim-single} recalls how the nonlocality of a single three-qubit state is tested with the Mermin and Svetlichny inequalities, and introduces the correlation tensor in which all our results are written. \Cref{sec:prelim-net} explains what a network is, what the classical model with independent sources looks like, and which network shapes appear in the paper. \Cref{sec:prelim-classes} fixes the quantum measurements we allow, so that it is always clear which number is being compared with which.

\subsection{Nonlocality of a single tripartite state}
\label{sec:prelim-single}

Consider a single source that sends one qubit each to three parties. Each party chooses one of two measurements, labeled by $x,y,z\in\{0,1\}$, and obtains an outcome $\pm1$. We write the corresponding observables as $A_x$, $B_y$, $C_z$. A Bell inequality bounds a combination of the averages $\avg{A_xB_yC_z}$ for all classical models, and quantum states that exceed the bound are called nonlocal. For three parties there is more than one classical model to exclude. In a fully local model each party answers using only a shared hidden variable and its own setting. In a weaker model, two of the parties may in addition share arbitrary resources, even nonlocal ones. Two standard inequalities test these two models. Consider the operators
\begin{align}
M  &= A_0B_0C_1+A_0B_1C_0+A_1B_0C_0-A_1B_1C_1, \label{eq:M}\\
M' &= A_1B_1C_0+A_1B_0C_1+A_0B_1C_1-A_0B_0C_0. \label{eq:Mp}
\end{align}
The operator $M'$ is obtained from $M$ by exchanging the settings $0$ and $1$ of all parties. The Mermin and Svetlichny inequalities read
\begin{align}
\abs{\avg{M}}\le 2, \qquad \abs{\avg{M+M'}}\le 4. \label{eq:MS}
\end{align}
The first holds for every fully local model~\cite{Mermin1990}. The second holds even when two of the parties may share arbitrary nonlocal resources~\cite{Svetlichny1987,Bancal2013}, so its violation shows that all three parties take part in the nonlocal correlations. Quantum mechanics reaches $4$ and $4\sqrt2$, both with the GHZ state.

All averages of products of Pauli observables of a three-qubit state are collected in its correlation tensor
\begin{align}
T_{\mu k\nu}=\Tr\!\left[\rho\,\sigma_\mu\otimes\sigma_k\otimes\sigma_\nu\right], \label{eq:T}
\end{align}
with $\sigma_0=\id$ and the Pauli matrices $\sigma_{1,2,3}$. For measurements along unit vectors $\vect a,\vect b,\vect c$ the three-body correlator is $T(\vect a,\vect b,\vect c)=\sum_{ikm}T_{ikm}a_ib_kc_m$. We denote by $\Me(\rho)$ the maximum of $\abs{\avg M}$ and by $\Sv(\rho)$ the maximum of $\abs{\avg{M+M'}}$ over all projective qubit measurements. Both refer to a single copy of $\rho$, and their local bounds are $2$ and $4$. In the networks below, the middle index of $T$ always belongs to the particle that is sent to an end node.

Our examples are the GHZ state, the generalized GHZ states and the $W$ state,
\begin{align}
\ket{\GHZ}&=\tfrac{1}{\sqrt2}\left(\ket{000}+\ket{111}\right), \nonumber\\
\ket{\GHZ_\alpha}&=\cos\alpha\ket{000}+\sin\alpha\ket{111}, \nonumber\\
\ket{W}&=\tfrac{1}{\sqrt3}\left(\ket{001}+\ket{010}+\ket{100}\right), \label{eq:states}
\end{align}
with $0<\alpha<\pi/2$ for the generalized GHZ states, and their noisy versions
\begin{align}
\rho_v=v\proj{\psi}+(1-v)\frac{\id}{8} \label{eq:noisy}
\end{align}
of visibility $v$. The noise multiplies every entry of $T$ except $T_{000}$ by $v$. The nonzero entries of the GHZ tensor are $T_{111}=1$, $T_{122}=T_{212}=T_{221}=-1$ and $T_{330}=T_{303}=T_{033}=1$.

\subsection{Networks and \texorpdfstring{$n$}{n}-locality}
\label{sec:prelim-net}

A network consists of several sources and several nodes. Each source emits particles to some of the nodes, and each node measures the particles it receives. In this paper the sources are independent, for instance because they are far apart or are switched on independently. A node that receives particles from several sources and measures them together is called a central node. A node that receives a single particle and chooses between two measurements is called an end node.

In the classical description, each source $i$ emits a hidden variable $\lambda_i$ with probability density $q_i(\lambda_i)$, and the variables of different sources are independent. Every node produces its outcome from its input, if it has one, and from the hidden variables of the sources it is connected to. For the closed network of \cref{fig:networks}(a), this gives
\begin{align}
p(a,b,c,d & |y,w)=\int\! \dd\lambda_1 \dd\lambda_2\,  q_1(\lambda_1)\,q_2(\lambda_2)\nonumber\\
&\times p(a|\lambda_1,\lambda_2)\,p(c|\lambda_1,\lambda_2)
p(b|y,\lambda_1)\,p(d|w,\lambda_2). \label{eq:bilocal}
\end{align}
Here $a$ and $c$ are the outcomes of the central nodes Alice and Charlie, who have no input and see both hidden variables. The outcomes $b,d=\pm1$ belong to the end nodes Bob and David, who see only the hidden variable of their own source and choose inputs $y,w\in\{0,1\}$. Correlations of this form are called bilocal, and with $n$ sources $n$-local. Correlations that cannot be written in this form are network nonlocal. In the quantum network each source emits a state $\rho_i$, the central nodes measure jointly the particles they hold, and the end nodes measure dichotomic observables.

Besides the closed and open networks of \cref{fig:networks}, one more shape will be useful, the star network~\cite{Tavakoli2014}. A star has a single central node that receives one particle from every source, while the other particle of each source goes to its own end node. If Alice and Charlie in \cref{fig:networks}(a) were replaced by one node holding all their particles, the closed network would become a star with two sources. We use the star only as a reference on the classical side, in \cref{sec:closed-ineq}. The networks studied in this paper are the closed and open ones.

\subsection{Measurement classes}
\label{sec:prelim-classes}

The value of a Bell expression depends on the measurements that are used. It can be evaluated for one fixed strategy, maximized over a class of strategies, or maximized over all quantum strategies. These are different numbers, and comparing them carelessly can be misleading, so we keep them apart throughout. Most of our quantum statements use the following class. The central nodes perform Bell-state measurements, possibly in rotated local frames, and the end nodes perform projective qubit measurements. A Bell-state measurement on two qubits returns two independent signs, the eigenvalues of the commuting parities $\sigma_3\otimes\sigma_3$ and $\sigma_1\otimes\sigma_1$, and the third parity follows from $\sigma_2\otimes\sigma_2=-(\sigma_1\otimes\sigma_1)(\sigma_3\otimes\sigma_3)$. In a frame given by orthonormal vectors $\{\vect e_1,\vect e_2,\vect e_3\}$, it returns the parities $\vect e_i\cdot\vect\sigma\otimes\vect e_i\cdot\vect\sigma$. \Cref{tab:classes} lists the quantities used in the paper together with the class in which each is computed, and whether it is known exactly.

\begin{table}[t]
\caption{\justifying\small Quantities used in the paper. All are normalized so that the local bound is $1$, except $\Me$ and $\Sv$, whose local bounds are $2$ and $4$. ``Exact'' means a proven maximum within the stated class.}
\label{tab:classes}
\begin{ruledtabular}
\begin{tabular}{lll}
Symbol & Central nodes & Status\\
\hline
$\Me$, $\Sv$ & \parbox[t]{0.3\columnwidth}{\raggedright none (one copy)} & \parbox[t]{0.36\columnwidth}{\raggedright exact for GHZ, numerical otherwise}\\[2pt]
$\Nc$ & \parbox[t]{0.3\columnwidth}{\raggedright Bell-state measurements} & \parbox[t]{0.36\columnwidth}{\raggedright closed form, \cref{eq:value}}\\[2pt]
$\Nc^{\rm prod}$ & \parbox[t]{0.3\columnwidth}{\raggedright product projective} & \parbox[t]{0.36\columnwidth}{\raggedright numerical lower bound}\\[2pt]
$\No$ & \parbox[t]{0.3\columnwidth}{\raggedright Bell-state measurement} & \parbox[t]{0.36\columnwidth}{\raggedright closed form, \cref{eq:No}}\\[2pt]
\cref{thm:parity} & \parbox[t]{0.3\columnwidth}{\raggedright arbitrary} & \parbox[t]{0.36\columnwidth}{\raggedright exact (tight bounds)}\\[2pt]
\end{tabular}
\end{ruledtabular}
\end{table}

\section{Closed network and Svetlichny nonlocality}
\label{sec:closed}

\subsection{The inequality}
\label{sec:closed-ineq}

We first ask which correlations the closed network can produce classically. Alice and Charlie are two different nodes, but they are connected to the same sources and therefore see the same hidden variables $\lambda_1,\dots,\lambda_n$. Two nodes that receive exactly the same information can coordinate their answers perfectly, so classically they should be no more powerful than a single node that produces both outcomes. If Alice and Charlie are merged into one node in this way, the closed network turns into the star network described in \cref{sec:prelim-net}. The following lemma states this precisely.

\begin{lemma}\label{lem:merge}
Replace Alice and Charlie by a single node that is connected to all sources and outputs the pair $(a,c)$, and keep all sources and end nodes unchanged, then a probability distribution $p(a,c,b^{(1)},\dots,b^{(n)}|y_1,\dots,y_n)$ can be produced by an $n$-local model of the closed network if and only if it can be produced by an $n$-local model of this merged network.
\end{lemma}

The proof in Appendix~\ref{app:ineq} rests on one observation. Whatever randomness a merged node uses to correlate $a$ and $c$ can be sent along with one of the sources, so that Alice and Charlie both receive it. The lemma does not say that the closed network and the star are the same network. They differ in the quantum case, which is our subject. It only says that their sets of classical correlations coincide. This has two consequences. First, any inequality that holds for the star also holds for the closed network, and inequalities for the star are well known~\cite{Tavakoli2014}. Second, every difference between the two networks comes from quantum mechanics. A merged quantum node could measure all its particles jointly, whereas Alice and Charlie can only perform a product of two measurements, one on each set of particles. The closed network asks what tripartite sources can achieve under this restriction.

Let $f_0$ and $f_1$ be any two functions of the outcomes $(a,c)$ of the central nodes with values $\pm1$, and let $b^{(i)}_{y_i}$ be the outcome of the $i$th end node for input $y_i$. We define
\begin{align}
I&=\frac{1}{2^n}\sum_{y_1\dots y_n}\avg{f_0\, b^{(1)}_{y_1}\cdots b^{(n)}_{y_n}}, \nonumber\\
J&=\frac{1}{2^n}\sum_{y_1\dots y_n}(-1)^{y_1+\dots+y_n}\avg{f_1\, b^{(1)}_{y_1}\cdots b^{(n)}_{y_n}}. \label{eq:IJ}
\end{align}
The functions $f_0$ and $f_1$ are fixed before the experiment. They simply turn the outcomes of the central nodes into two signs, one for each correlator.

\begin{theorem}\label{thm:ineq}
Every $n$-local correlation of the closed network satisfies
\begin{align}
\abs{I}^{1/n}+\abs{J}^{1/n}\le 1. \label{eq:ineq}
\end{align}
\end{theorem}

For $n=2$ this is the bilocality inequality $\sqrt{\abs I}+\sqrt{\abs J}\le1$ of Ref.~\cite{Branciard2012}, written for central nodes without input. A short proof is given in Appendix \ref{app:ineq}.

\subsection{Quantum value with Bell-state measurements}
\label{sec:closed-value}

We now take two identical sources and let both central nodes perform Bell-state measurements, Alice in the frame $\{\vect e_i\}$ and Charlie in the frame $\{\vect f_j\}$. We write $A_i=\vect e_i\cdot\vect\sigma$ and $C_j=\vect f_j\cdot\vect\sigma$, and set $A_0=C_0=\id$. The index $0$ means that this central node is ignored. As $f_0$ we take the product of Alice's parity $A_\mu\otimes A_\mu$ with Charlie's parity $C_\nu\otimes C_\nu$, and as $f_1$ the analogous product with indices $(\mu',\nu')$.

With this choice, $I$ factorizes into one term per source, because every operator in it is a product of a part acting on source 1 and a part acting on source 2. Each term contains the vector
\begin{align}
\vect\tau_{\mu\nu}=\left(\Tr\!\left[\rho\,A_\mu\otimes\sigma_k\otimes C_\nu\right]\right)_{k=1,2,3}. \label{eq:tau}
\end{align}
It describes the particle of the end node as seen through the observables $A_\mu$ and $C_\nu$ of the central nodes, much like a Bloch vector. Its length never exceeds one.

The end node turns its two settings $\vect b_0,\vect b_1$ into two perpendicular directions, $\vect b_0+\vect b_1\propto\vect n$ for $I$ and $\vect b_0-\vect b_1\propto\vect n'$ for $J$. The network value therefore depends on how much of $\vect\tau_{\mu\nu}$ and $\vect\tau_{\mu'\nu'}$ two perpendicular directions can capture. For two vectors $\vect\tau,\vect\tau'$ this is
\begin{align}
\Lambda(\vect\tau,\vect\tau')&=\max_{\vect n\perp\vect n'}\left[(\vect\tau\cdot\vect n)^2+(\vect\tau'\cdot\vect n')^2\right]\nonumber\\
&=\frac{1}{2}\left[s+\sqrt{s^2-4\left(\vect\tau\cdot\vect\tau'\right)^2}\,\right], \label{eq:Lambda}
\end{align}
with $s=\abs{\vect\tau}^2+\abs{\vect\tau'}^2$. It equals $\abs{\vect\tau}^2+\abs{\vect\tau'}^2$ when the two vectors are orthogonal, and the larger of the two squared lengths when they are parallel.

\begin{proposition}\label{prop:value}
For two copies of a three-qubit state $\rho$, Bell-state measurements at the central nodes and projective measurements at the end nodes, the largest value of $\sqrt{\abs I}+\sqrt{\abs J}$ is $\Nc(\rho)$, with
\begin{align}
\Nc(\rho)^2=\max_{\{\vect e_i\},\{\vect f_j\}}\ \max_{\mu\nu,\,\mu'\nu'}\ \Lambda\!\left(\vect\tau_{\mu\nu},\vect\tau_{\mu'\nu'}\right). \label{eq:value}
\end{align}
The maxima run over the frames of the two central nodes and over the indices $\mu,\nu,\mu',\nu'\in\{0,1,2,3\}$.
\end{proposition}

The proof is in Appendix~\ref{app:value}. Since every $\vect\tau_{\mu\nu}$ has length at most one, $\Lambda\le2$ and $\Nc\le\sqrt2$. If Charlie is ignored ($\nu=\nu'=0$), the vectors $\vect\tau_{\mu0}$ are the rows of the correlation matrix of the reduced state of Alice and the end node, and \cref{eq:value} returns half of its CHSH value~\cite{Horodecki1995,Gisin2017}. We denote this single-node value by $\mathcal N_A$, and by $\mathcal N_C$ the value with Alice ignored.

\subsection{Svetlichny nonlocality is sufficient}
\label{sec:closed-svet}

The closed-network value is built from pairs of three-body correlators in which the end node uses perpendicular directions. Once the Svetlichny operator is sorted by the settings of one party, it has the same structure. This observation gives our first main result.

\begin{theorem}\label{thm:svet}
For every three-qubit state $\rho$, the two-copy closed-network value with Bell-state measurements is at least the normalized single-copy Svetlichny value,
\begin{align}
\Nc(\rho)\ \ge\ \frac{\Sv(\rho)}{4}. \label{eq:svet}
\end{align}
Hence every state that violates the Svetlichny inequality violates the closed-network inequality~\eqref{eq:ineq} with two copies.
\end{theorem}

The proof is in Appendix~\ref{app:svet}. The reason why the network can do better is simple. Sorted by the two perpendicular directions of the end node, the Svetlichny expression contains four three-body correlators $X,Y,Z,W$, two for each direction. Alice and Charlie must switch between two settings, and as a result the correlators enter with weights whose absolute values add up to at most one. Each bracket is therefore a weighted average of two correlators, and the Svetlichny value is bounded by $\max(X^2,Y^2)+\max(Z^2,W^2)$. A central node has no settings to switch. Its Bell-state measurement delivers all parities in every run, so the network can select the larger correlator for each direction and reach this bound.

The theorem says that Svetlichny nonlocality is sufficient for a violation. The examples below show that it is far from necessary.

\subsection{Examples}
\label{sec:closed-examples}

\textit{GHZ states.} For the noisy GHZ state all vectors $\vect\tau_{\mu\nu}$ are multiplied by $v$. Take $f_0=A_3A_3$, Alice's $\sigma_3\otimes\sigma_3$ parity, with $\vect\tau_{30}=v\,\hat{\vect z}$, and $f_1=A_1A_1\,C_1C_1$, the product of the $\sigma_1\otimes\sigma_1$ parities of both central nodes, with $\vect\tau_{11}=v\,\hat{\vect x}$. The two vectors are orthogonal, and Bob and David measure $(\sigma_3\pm\sigma_1)/\sqrt2$. This gives $\sqrt2\,v$, and \cref{eq:value} cannot exceed it because both vectors already have the maximal length $v$. Therefore
\begin{align}
\Nc(\rho^{\GHZ}_v)=\sqrt2\,v=\frac{\Sv(\rho^{\GHZ}_v)}{4}, \label{eq:ghz}
\end{align}
and the bound of \cref{thm:svet} is saturated. The closed network and the Svetlichny inequality are both violated for $v>1/\sqrt2$, whereas the Mermin inequality is violated already for $v>1/2$. A numerical search over all product projective measurements of the central nodes also returned $\sqrt2\,v$ (Appendix~\ref{app:numerics}), which suggests that Bell-state measurements are optimal for GHZ sources.

The first correlator above uses Alice alone, and the second uses both central nodes. No strategy can do with only one of them. The reduced two-qubit states of the GHZ state are separable, and separable sources cannot violate a network inequality, whatever the other central node measures. Hence $\mathcal N_A,\mathcal N_C\le1$, and the nonlocality of the closed GHZ network needs both central nodes.

\textit{Generalized GHZ states.} Here $\vect\tau_{30}=\hat{\vect z}$ for every $\alpha$, while $\vect\tau_{11}=\sin2\alpha\,\hat{\vect x}$. The same strategy then gives the following.

\begin{corollary}\label{cor:gghz}
Every generalized GHZ state with $0<\alpha<\pi/2$ violates the closed-network inequality with two copies,
\begin{align}
\Nc(\GHZ_\alpha)\ \ge\ \sqrt{1+\sin^22\alpha}\ >\ 1. \label{eq:gghz}
\end{align}
\end{corollary}

In contrast, the Svetlichny inequality is violated only for $\sin2\alpha>1/\sqrt2$~\cite{Ghose2009}, and the Mermin inequality only for $\sin2\alpha>1/2$~\cite{Scarani2001,Zukowski2002}. The difference is qualitative. As $\alpha$ approaches $0$ or $\pi/2$ the state approaches a product state, yet for every value in between the network inequality is violated. A numerical optimization of \cref{eq:value} over the frames returns the right-hand side of \cref{eq:gghz}, so we expect it to be the exact value, but we have not proven this. \Cref{fig:hierarchy}(a) compares the network value with the Svetlichny and Mermin values.

\textit{$W$ states.} For the $W$ state, $\vect\tau_{33}=-\hat{\vect z}$ and $\vect\tau_{13}=\tfrac23\hat{\vect x}$ are orthogonal, which gives $\Nc\ge\sqrt{13}/3\approx1.202$. The numerical optimization over frames finds no larger value. The Svetlichny value is smaller, $\Sv(W)/4\approx1.089$, so the network inequality tolerates more noise. Its critical visibility is $0.832$, compared with $0.919$ for the Svetlichny inequality. Unlike for GHZ states, Bell-state measurements are not the best choice for the $W$ state. A numerical search over product projective measurements of the two central nodes finds $\Nc^{\rm prod}\ge1.316$. Since a numerical search can only find lower bounds on a maximum, the corresponding critical visibility $0.760$ is an upper bound on the true threshold within this class.

The critical visibilities quoted here follow from a simple fact. For a fixed strategy, every network value scales linearly with $v$ (Appendix~\ref{app:numerics}), so a pure-state value $K$ gives the threshold $v_c=1/K$. \Cref{tab:values} collects the values, and \cref{fig:hierarchy}(b) shows the resulting ranges of visibility.

\begin{table}[t]
\caption{\justifying\small Values for the GHZ and $W$ states in units of the local bound. The critical visibility of the noisy state is the inverse of each entry. Entries marked $^{a}$ are analytical. Entries marked $^{n}$ are numerical maxima, which are lower bounds on the true maxima. For $\Nc(W)$ the value $\sqrt{13}/3$ is a proven lower bound, and the numerics find no larger value.}
\label{tab:values}
\begin{ruledtabular}
\begin{tabular}{lccccc}
State & $\Sv/4$ & $\Nc$ & $\Nc^{\rm prod}$ & $\No$ & $\Me/2$\\
\hline
GHZ & $1.414^{a}$ & $1.414^{a}$ & $1.414^{n}$ & $2.000^{a}$ & $2.000^{a}$\\
$W$ & $1.089^{n}$ & $1.202^{a}$ & $1.316^{n}$ & $1.521^{n}$ & $1.523^{n}$\\
\end{tabular}
\end{ruledtabular}
\end{table}

\textit{Biseparable states.} The converse of \cref{thm:svet} fails in the strongest possible way. Consider a Bell pair shared by Alice and the end node, with Charlie's particle in a product state,
\begin{align}
\ket{\Phi^+}_{AB}\otimes\ket{0}_C. \label{eq:bisep}
\end{align}
This state is biseparable, and $\Sv/4=1$. Yet $\vect\tau_{30}=\hat{\vect z}$ and $\vect\tau_{10}=\hat{\vect x}$ give $\Nc=\sqrt2$, because Alice alone runs the bilocal scenario of Ref.~\cite{Branciard2012}. A violation of the closed-network inequality therefore certifies neither genuine tripartite entanglement nor genuine tripartite nonlocality of the sources, in line with the observation of Ref.~\cite{Mukherjee2017}. Conversely, if the entangled pair is shared by Alice and Charlie and the end node holds the unentangled particle, all vectors $\vect\tau_{\mu\nu}$ are parallel to the Bloch vector of that particle and $\Nc=1$. Both observations agree with the numerical findings of Ref.~\cite{Mukherjee2020} for the closed network with three sources. What sets the GHZ and $W$ states apart from this example is that for them $\mathcal N_A,\mathcal N_C\le1$, so their violation cannot be produced by a single central node.

These examples settle how the closed-network inequality relates to the two single-copy inequalities. A Svetlichny violation implies a network violation, but not conversely. A Mermin violation neither implies nor is implied by a violation with Bell-state measurements, as the noisy GHZ state with $1/2<v\le1/\sqrt2$ and the generalized GHZ states with $\sin2\alpha\le1/2$ show. \Cref{fig:random}(a) illustrates \cref{thm:svet} on random states.

\begin{figure*}[t]
\centering
\includegraphics[width=0.94\textwidth]{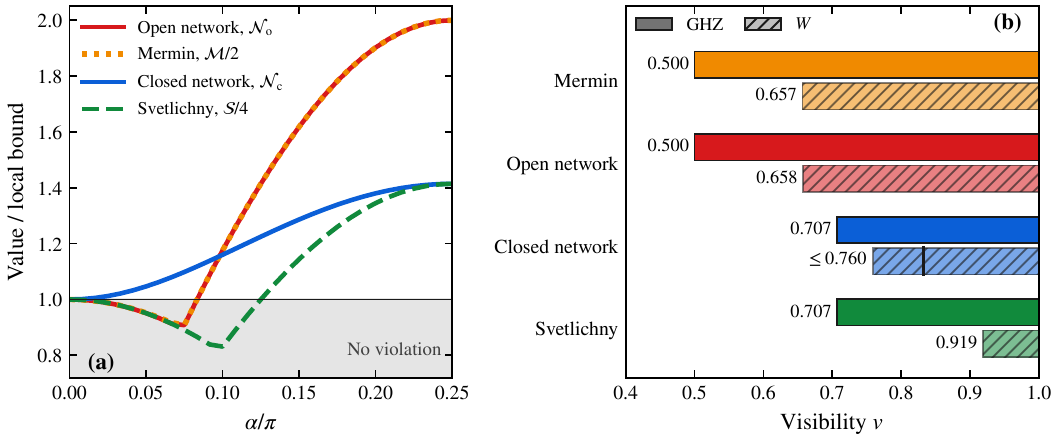}
\caption{\justifying\small (a) Generalized GHZ states. Closed-network value $\sqrt{1+\sin^22\alpha}$ of \cref{eq:gghz} (blue solid line), Svetlichny value $\Sv/4$ (green dashed line), open-network value $\No$ (red solid line) and Mermin value $\Me/2$ (orange dotted line), all in units of the local bound. The red and orange curves coincide wherever they exceed $1$. Values in the gray region do not violate the corresponding inequality. (b) Ranges of visibility in which noisy GHZ states (solid bars) and noisy $W$ states (hatched bars) violate each inequality. The number at the left end of each bar is the critical visibility. For the closed network, the hatched blue bar uses the numerical value $\Nc^{\rm prod}$, so its left end is an upper bound on the threshold, and the black tick marks the threshold $0.832$ for Bell-state measurements.}
\label{fig:hierarchy}
\end{figure*}

\begin{figure*}[t]
\centering
\includegraphics[width=0.94\textwidth]{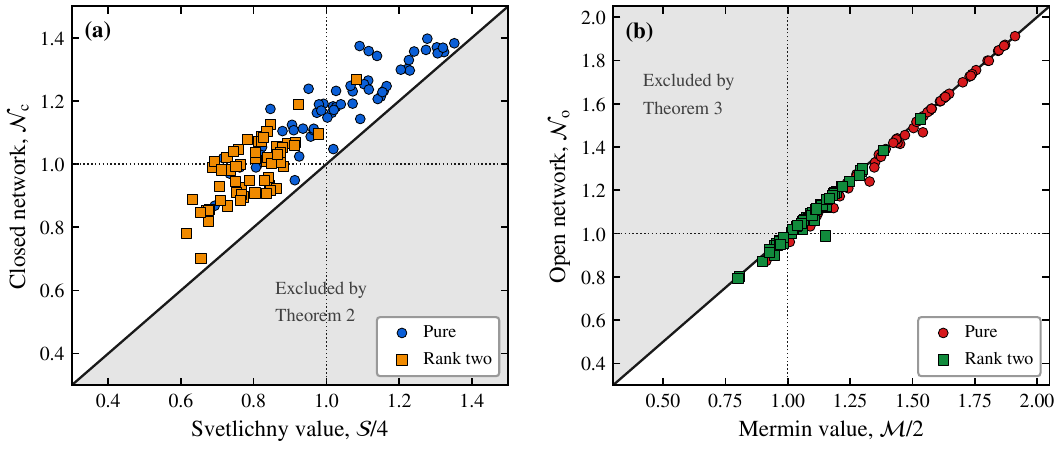}
\caption{\justifying\small Network values against single-copy values for 120 random three-qubit states (sampling described in Appendix~\ref{app:numerics}). Blue circles and red circles are pure states, orange squares and green squares are states of rank two. (a) Closed-network value $\Nc$ against $\Sv/4$. All points lie above the black diagonal, as required by \cref{thm:svet}. (b) Open-network value $\No$ against $\Me/2$. All points lie on or below the diagonal, as required by \cref{thm:mermin}. The plot is a consistency check of the theorems, not evidence for them.}
\label{fig:random}
\end{figure*}

\section{The swapped state at the end nodes}
\label{sec:swap}

We have just seen that for GHZ sources the nonlocality of the closed network needs both central nodes. A natural next question is whether the same holds for the state that the network prepares. In the bilocal scenario, a network violation goes together with entanglement swapping. The central node's measurement leaves the end nodes in an entangled state, although their particles never interacted~\cite{Zukowski1993,Branciard2010}. With tripartite sources the swapping involves both central nodes~\cite{Bose1998}.

Let both central nodes perform Bell-state measurements in the standard frame. We label Alice's outcome by the signs $a_z$ and $a_x$, the eigenvalues of $\sigma_3\otimes\sigma_3$ and $\sigma_1\otimes\sigma_1$, and Charlie's outcome by $c_z$ and $c_x$.

\begin{proposition}\label{prop:swap}
For two GHZ sources the outcomes satisfy $c_z=a_z$ with certainty, and Bob and David are left in the Bell state with
\begin{align}
\sigma_3\otimes\sigma_3=a_z, \qquad \sigma_1\otimes\sigma_1=a_xc_x. \label{eq:swap}
\end{align}
Conditioned on the outcome of only one central node, the state of Bob and David is separable. For $n$ GHZ sources and central nodes that measure in the $n$-qubit GHZ basis, the end nodes are left in an $n$-qubit GHZ state whose relative sign is the product of the signs found by the two central nodes, and again the outcome of one central node leaves them in a separable state.
\end{proposition}

The proof in Appendix~\ref{app:swap} uses only the stabilizers of the GHZ state. The proposition makes the observation of the previous section concrete. Alice's outcome tells Bob and David that their outcomes in the $\sigma_3$ basis are perfectly correlated or anticorrelated. The relative sign of the entangled state, however, is the product $a_xc_x$ of one sign from each central node. Without Charlie's sign, the two possible Bell states are equally likely, and their mixture is separable. In this precise sense the swapped entanglement is heralded only by the joint outcome. The way a secret is split among parties in GHZ-based secret sharing~\cite{Hillery1999} is similar, but we make no information-theoretic claim here.

The $W$ state behaves differently. Two $W$ sources leave Bob and David in a Bell state with probability $4/9$, in a state of concurrence $4/5$ with probability $5/18$, and in a separable state otherwise, with an average concurrence of $2/3$. Here the outcome of one central node already swaps some entanglement, because the reduced states of the $W$ state are entangled.

\textit{Noise.} With noisy GHZ sources the certain event $c_z=a_z$ becomes a herald. It occurs with probability $(1+v^2)/2$ and leaves Bob and David in a Werner state~\cite{Werner1989} of visibility
\begin{align}
w=\frac{2v^2}{1+v^2}. \label{eq:w}
\end{align}
The complementary event leaves them maximally mixed (Appendix~\ref{app:swap}). A Werner state is entangled for $w>1/3$ and violates the CHSH inequality for $w>1/\sqrt2$, so the heralded state is entangled for $v>1/\sqrt5\approx0.447$ and violates the CHSH inequality for $v>(2\sqrt2-1)^{-1/2}\approx0.740$.

Compare this with the closed-network inequality, which is violated for $v>1/\sqrt2\approx0.707$. In the window $0.707<v\le0.740$ the network correlations violate the network inequality, while the heralded state of Bob and David does not violate the CHSH inequality. The network violation is therefore not inherited from a CHSH violation of the swapped state. The window is narrower than for bipartite sources. There, two Werner states of visibility $v$ swap to visibility $v^2$, and the CHSH threshold is $2^{-1/4}\approx0.841$. The herald $c_z=a_z$ discards some of the noise, as \cref{fig:swap} shows.

\begin{figure}[t]
\centering
\includegraphics[width=0.9\columnwidth]{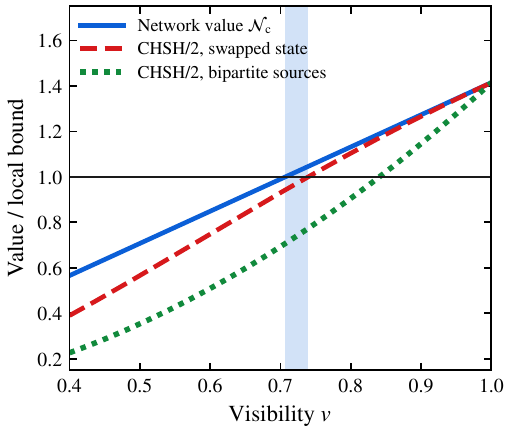}
\caption{\justifying\small Noisy GHZ sources of visibility $v$. Closed-network value $\Nc=\sqrt2\,v$ (blue solid line), half the CHSH value of the heralded swapped state of \cref{eq:w} (red dashed line), and half the CHSH value of the state swapped from two bipartite Werner states (green dotted line). In the shaded blue window the network inequality is violated while the heralded state does not violate the CHSH inequality.}
\label{fig:swap}
\end{figure}

\section{Open network and Mermin nonlocality}
\label{sec:open}

The closed network compared favorably with the Svetlichny test because its central nodes read all parities at once. We now change the wiring. In the open network of \cref{fig:networks}(c) a single central node, Charlie, holds one particle from each source, and the other two particles of every source go to two end nodes. Alice and Bob measure the particles of the first source with settings $x$ and $y$, and David and Eve those of the second source with settings $w$ and $z$. Each source now has two end nodes and only one particle at the central node. As we shall see, this reverses the relation found above.

The two end nodes of a source can be correlated through its hidden variable, so the single-node combinations used in \cref{eq:IJ} must be replaced by combinations of two end nodes. We use the two halves of the Mermin operator,
\begin{align}
U_0=\tfrac12\left(A_0B_0-A_1B_1\right), \qquad U_1=\tfrac12\left(A_0B_1+A_1B_0\right), \label{eq:kernel}
\end{align}
and the same combinations $V_0,V_1$ for David and Eve. With these, $M=2(U_0C_1+U_1C_0)$. In the Mermin test the third party couples $U_0$ to one setting and $U_1$ to the other. In the network, the central node couples them to two signs $g_0,g_1=\pm1$ computed from its single outcome, and we define
\begin{align}
I_{\rm o}=\avg{g_0\,U_0V_0}, \qquad J_{\rm o}=\avg{g_1\,U_1V_1}. \label{eq:IJopen}
\end{align}

\begin{proposition}\label{prop:open}
Every bilocal correlation of the open network satisfies
\begin{align}
\sqrt{\abs{I_{\rm o}}}+\sqrt{\abs{J_{\rm o}}}\le1, \label{eq:open}
\end{align}
and with $n$ sources the square roots are replaced by $n$th roots. The algebraic maximum of the left-hand side is $2$, and two GHZ sources reach it.
\end{proposition}

The local bound is proven in Appendix~\ref{app:ineq}. For GHZ sources, let Charlie perform a Bell-state measurement and set $g_0$ and $g_1$ to the parities of $\sigma_1\otimes\sigma_1$ and $\sigma_2\otimes\sigma_2$, while all end nodes measure $\sigma_1$ for setting $0$ and $\sigma_2$ for setting $1$. The GHZ correlators give $\avg{U_0\otimes\sigma_1}=1$ and $\avg{U_1\otimes\sigma_2}=-1$ for each source, so that $I_{\rm o}=J_{\rm o}=1$. No bilocal model can make both correlators reach their maximal size, which is the network form of the GHZ paradox~\cite{GHSZ1990,Mermin1990b}. The violation is a factor of two, larger than the factor $\sqrt2$ in the closed network and in the bilocal scenario with bipartite sources.

To compare with the Mermin value, fix the settings of the end nodes and define the vectors $\vect u_0$ and $\vect u_1$ with components $\Tr[\rho\,U_0\otimes\sigma_k]$ and $\Tr[\rho\,U_1\otimes\sigma_k]$. They play the role of $\vect\tau$ in the closed network. Let $\No(\rho)$ be the maximum of the left-hand side of \cref{eq:open} for two copies of $\rho$, when Charlie performs a Bell-state measurement in an arbitrary frame, $g_0$ and $g_1$ are two of its parities, and the end nodes measure projectively.

\begin{theorem}\label{thm:mermin}
For every three-qubit state $\rho$,
\begin{align}
\No(\rho)^2&=\max\left[\abs{\vect u_0}^2+\abs{\vect u_1}^2+2\abs{\vect u_0\times\vect u_1}\right], \label{eq:No}\\
\No(\rho)&\le\frac{\Me(\rho)}{2}, \label{eq:mermin}
\end{align}
where the maximum is over the settings of the end nodes. Equality holds in \cref{eq:mermin} if $\vect u_0\perp\vect u_1$ at the settings that maximize the Mermin value. Hence, if two copies of $\rho$ violate \cref{eq:open} with a Bell-state measurement at the central node, one copy of $\rho$ violates the Mermin inequality.
\end{theorem}

The proof is in Appendix~\ref{app:mermin}. The reason for the reversal is geometric. In the Mermin test the third party can align its two settings freely, one with $\vect u_0$ and one with $\vect u_1$, and obtains $\abs{\vect u_0}+\abs{\vect u_1}$. The two parities of a Bell-state measurement belong to perpendicular directions. When $\vect u_0$ and $\vect u_1$ are not perpendicular, the central node cannot align with both, and it loses. In the closed network the restriction to perpendicular directions sat at the end nodes, which suffer from it in the Svetlichny test as well. Here it sits at the central node, which does not suffer from it in the Mermin test.

For the generalized GHZ states in the range where the Mermin inequality is violated, the optimal $\vect u_0$ and $\vect u_1$ are perpendicular, so nothing is lost and $\No=\Me/2=2\sin2\alpha$. For the noisy GHZ state $\No=\Me/2=2v$, so the threshold is $v>1/2$, the same as for the Mermin inequality. For the $W$ state the two vectors are not perpendicular and $\No\approx1.521$ is slightly below $\Me/2\approx1.523$. \Cref{fig:random}(b) shows the same behavior for random states.

The converse of \cref{thm:mermin} fails as well. The state of \cref{eq:bisep}, with its unentangled particle at the central node, has $\Me/2=\sqrt2$ but $\No=1$, because a central node holding an uncorrelated particle cannot herald anything. If instead the unentangled particle goes to an end node, $\No=\sqrt2$. As in the closed network, a violation does not require genuine tripartite entanglement.

Taken together, and for the inequalities and measurements considered here, the two wirings are bounded by the two single-copy inequalities from opposite sides,
\begin{align}
\Sv(\rho)>4\ &\Rightarrow\ \Nc(\rho)>1, \nonumber\\
\No(\rho)>1\ &\Rightarrow\ \Me(\rho)>2. \label{eq:hierarchy}
\end{align}
Neither implication can be reversed. Whether other notions of multipartite nonlocality~\cite{Bancal2013} or other network inequalities obey similar relations is left open.

\section{Many sources and the parity effect}
\label{sec:nlocal}

So far we have used two sources. In both networks the central nodes have no input, and this has a consequence that becomes visible once more sources are added. A central node must provide the signs for both correlators, $I$ and $J$, from the same measurement in every run. We now show that this is possible without loss for an even number of sources, but not for an odd number.

\textit{Intuition.} Consider the closed network with $n$ sources, as in \cref{fig:networks}(b), and qubit end nodes. The sum and the difference of the two observables of the $i$th end node are proportional to Pauli operators along perpendicular directions $\vect n_i$ and $\vect n_i'$, and two such operators anticommute. The correlator $I$ probes the product $P$ of the $n$ operators $\vect n_i\cdot\vect\sigma$, and $J$ probes the product $Q$ of the $n$ operators $\vect n_i'\cdot\vect\sigma$. Moving $Q$ past $P$ produces one minus sign per factor, so $PQ=(-1)^nQP$. For even $n$ the two products commute, and the central nodes can herald states in which both have definite values. For odd $n$ they anticommute, and no state of the end nodes can have large expectation values of both~\cite{Kurzynski2011}. A central node with a binary input could measure in a basis adapted to $P$ for one input and to $Q$ for the other, and the star inequality with such a node is violated up to $\sqrt2$ for every $n$~\cite{Tavakoli2014,Andreoli2017}. A central node without input has to compromise, as long as the inequality contains only these two correlators.

\begin{theorem}\label{thm:parity}
Consider $n$ sources that emit arbitrary states, and central nodes without input that perform arbitrary measurements.

(i) In the closed network with qubit end nodes and projective measurements,
\begin{align}
\abs I^{1/n}+\abs J^{1/n}\le
\begin{cases}
\sqrt2, & n\ \text{even},\\
2^{\frac12-\frac1{2n}}, & n\ \text{odd}.
\end{cases} \label{eq:parityc}
\end{align}

(ii) In the open network with projective dichotomic measurements of the end nodes in any dimension,
\begin{align}
\abs{I_{\rm o}}^{1/n}+\abs{J_{\rm o}}^{1/n}\le
\begin{cases}
2, & n\ \text{even},\\
2^{1-\frac1{2n}}, & n\ \text{odd}.
\end{cases} \label{eq:parityo}
\end{align}

All four bounds are attained with GHZ sources.
\end{theorem}

The proof is in Appendix~\ref{app:parity}. Four remarks clarify its scope.

First, the bounds are statements about the maximal quantum values of the particular expressions in \cref{eq:ineq,eq:open}. They do not limit network nonlocality in general.

Second, the proof uses neither the form of the source states nor their independence. Independence is what makes the local bound equal to $1$. The quantum upper bound only uses the algebra of the end-node observables, so it holds for any sources. For the same reason the effect is not caused by the sources being tripartite. It holds for any star whose central nodes have no input and whose inequality has the two-correlator form of \cref{eq:ineq}, and it does not disappear when several central nodes share the work. In particular, the trilocal inequalities of Ref.~\cite{Mukherjee2020} for three tripartite sources have this form, with particular choices of $f_0$ and $f_1$, so their quantum violation can never exceed $2^{1/3}$. The effect is a property of two-correlator inequalities and not of central nodes without input as such. Tavakoli \textit{et al.}~\cite{Tavakoli2014} constructed, for a central node that performs a single GHZ-basis measurement, a star inequality with $2^{n-1}$ correlators, one for each stabilizer of the GHZ state with an even number of $\sigma_2$ factors, and showed that it is violated by a factor $\sqrt2$ for every $n$. Using more correlators thus removes the parity effect. To our knowledge, the tight bounds of \cref{thm:parity} for two-correlator inequalities have not been stated before.

Third, part (i) assumes qubit end nodes with projective measurements, whereas part (ii) holds in any dimension. The difference comes from the closed network, where we use that the sum and the difference of two qubit observables are proportional to Pauli operators.

Fourth, for $n=1$ the closed bound equals the local bound $1$. A single source with central nodes that have no input cannot violate \cref{eq:ineq}. The open bound for $n=1$ is $\sqrt2$, since a single tripartite source with a fixed measurement at one party can still violate a Bell inequality.

For three sources the maximal closed-network value is $2^{1/3}\approx1.26$, below the value $\sqrt2\approx1.41$ for two sources. Adding a third source to a maximally violating pair lowers the best violation. The odd values approach the even ones only as $n$ grows, as shown in \cref{fig:parity}(a).

\begin{figure*}[t]
\centering
\includegraphics[width=0.94\textwidth]{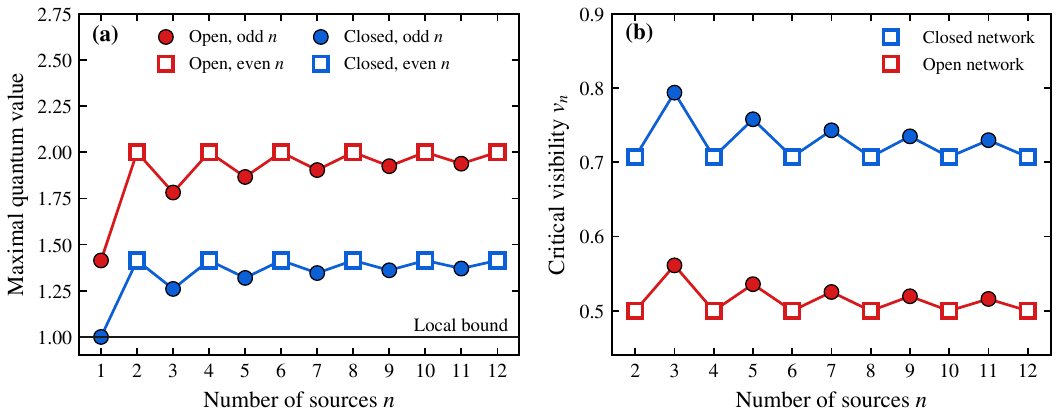}
\caption{\justifying\small Parity effect. (a) Maximal quantum values of the closed-network inequality (blue) and the open-network inequality (red) against the number of sources, from \cref{thm:parity}. Filled circles are odd $n$, open squares are even $n$, and the black line is the local bound. (b) Critical visibility of noisy GHZ sources for the strategies that attain these maxima, with the same colors and markers.}
\label{fig:parity}
\end{figure*}

\textit{Optimal strategies.} For even $n$ the sources are grouped in pairs, and each pair is treated as in \cref{sec:closed-examples,sec:open}. For odd $n$ one source is left over. In the closed network both central nodes measure $\sigma_1$ on its particles and multiply both $f_0$ and $f_1$ by the results, and its end node measures $\sigma_1$ and $\sigma_2$. For noisy GHZ sources these strategies give $\sqrt2\,v$ for even $n$ and $2^{1/2-1/(2n)}v$ for odd $n$ in the closed network, and $2v$ and $2^{1-1/(2n)}v$ in the open network. The resulting critical visibilities are plotted in \cref{fig:parity}(b).

\textit{Three sources and tripartite swapping.} For $n=3$ the natural measurement of each central node is the three-qubit GHZ basis, which by \cref{prop:swap} leaves the three end nodes in a GHZ state. This measurement also attains the bound of \cref{thm:parity}. It fixes the values of $\sigma_1\otimes\sigma_1\otimes\sigma_1$ and $\sigma_2\otimes\sigma_2\otimes\sigma_1$, which commute because they differ in two positions. Let $f_0$ be the product of the first operator at both central nodes, and $f_1$ the product of the second operator at Alice and the first at Charlie. With settings $(\sigma_1\pm\sigma_2)/\sqrt2$ for $B_1$ and $B_2$ and $\sigma_1,\sigma_2$ for $B_3$, one finds $I=J=1/4$ and the value $2^{1/3}$.

With noisy GHZ sources, the two central nodes find matching GHZ-basis outcomes with probability $(1+3v^2)/4$ and then leave the end nodes in a state that is diagonal in the GHZ basis. Its Mermin and Svetlichny values are $16v^3/(1+3v^2)$ and $16\sqrt2\,v^3/(1+3v^2)$, so it violates the Mermin inequality for $v>0.661$ and the Svetlichny inequality for $v>0.804$. The network strategy above violates \cref{eq:ineq} for $v>2^{-1/3}\approx0.794$. As for two sources, there is a window, $0.794<v\le0.804$, in which the network inequality is violated while the heralded three-qubit state does not violate the Svetlichny inequality.

\section{Summary and conclusions}
\label{sec:conclusions}

In this work, we studied how the nonlocality of a network whose sources emit three-qubit states is related to the Mermin and Svetlichny nonlocality of a single source. For the inequalities and measurements studied here, the answer depends on the wiring. In the closed network, where two central nodes without input share every source, the two-copy value with Bell-state measurements is bounded from below by the normalized Svetlichny value. Svetlichny nonlocality is thus sufficient for a violation, while all entangled generalized GHZ states, noisy $W$ states and even biseparable states show that it is not necessary. In the open network, where one central node is shared, the two-copy value with a Bell-state measurement is bounded from above by the normalized Mermin value, and GHZ sources reach the algebraic maximum.

Both bounds have a simple operational origin. A central node without input does not switch settings. In the closed network this is an advantage, because the central nodes read all parities of their Bell-state measurements in every run, whereas the parties of a Svetlichny test must switch. In the open network it is a limitation, because the two parities that the central node uses belong to perpendicular directions, whereas the third party of a Mermin test may choose any two directions. With many sources the same feature produces the parity effect. A central node without input cannot decide in each run which of two anticommuting observables to herald, and for an odd number of sources this lowers the maximal violation from $\sqrt2$ to $2^{1/2-1/(2n)}$ in the closed network and from $2$ to $2^{1-1/(2n)}$ in the open network.

For entanglement distribution the closed network has a further feature. With GHZ sources, the entanglement swapped to the end nodes is heralded only by the joint outcome of the two central nodes, while the outcome of one central node leaves the end nodes in a separable state. In the presence of noise, the network inequality is violated at visibilities where the heralded state does not violate the CHSH inequality for two sources, or the Svetlichny inequality for three. These statements concern the specific tests named here, and we do not claim that the network offers a practical advantage for a particular task.

Several directions remain open. The values $\Nc$ and $\No$ refer to Bell-state measurements. Our numerics indicate that product projective measurements at the central nodes do not help for GHZ sources but do help for $W$ sources, and a closed form for the optimal measurement of a general state would be valuable. It would also be interesting to know whether the parity bound of the closed network holds for end nodes of higher dimension, and whether inequalities with many correlators, like the one of Ref.~\cite{Tavakoli2014}, can be built for the closed network and related to the Mermin and Svetlichny values of a single source. Since biseparable sources already violate the inequalities used here, a further task is to find network tests with tripartite sources that certify the genuine multipartite character of the sources.

\appendix

\section{Local bounds: proofs of Lemma~\ref{lem:merge}, Theorem~\ref{thm:ineq} and Proposition~\ref{prop:open}} \label{app:ineq}

\textit{Proof of \cref{lem:merge}.} A model of the closed network is a model of the merged network in which the response of the merged node happens to factorize, $p(a,c|\vect\lambda)=p(a|\vect\lambda)\,p(c|\vect\lambda)$. So every closed-network model is a model of the merged network. For the converse, take a model of the merged network with response $p(a,c|\lambda_1,\dots,\lambda_n)$. Any such response can be produced by first drawing a random number $\mu$, independent of everything else, and then computing a deterministic function,
\begin{align}
(a,c)=F(\lambda_1,\dots,\lambda_n,\mu). \label{eq:Fdet}
\end{align}
Now let the first source emit the pair $\lambda_1'=(\lambda_1,\mu)$ instead of $\lambda_1$. The sources are still independent, because $\mu$ is independent of all other variables. Both central nodes receive $\lambda_1'$, so both know $\mu$. Alice outputs the first component of $F$, Charlie outputs the second, and the first end node ignores $\mu$. This closed-network model reproduces the model of the merged network exactly. \qed

\textit{Proof of \cref{thm:ineq}.} By \cref{lem:merge} we may use one central node. Let $\phi_j(\vect\lambda)\in[-1,1]$ be the mean value of $f_j$ for given $\vect\lambda=(\lambda_1,\dots,\lambda_n)$, and $\beta^{(i)}_y(\lambda_i)\in[-1,1]$ the mean outcome of the $i$th end node for input $y$.

\textit{Step 1.} The sums over inputs in \cref{eq:IJ} factorize. For example $\sum_{y,w}b_yd_w=(b_0+b_1)(d_0+d_1)$. Inserting the model, we get
\begin{align}
I=\int\!\dd\vect\lambda\,\prod_{i}q_i(\lambda_i)\ \phi_0(\vect\lambda)\prod_{i}\frac{\beta^{(i)}_0+\beta^{(i)}_1}{2}. \label{eq:Imodel}
\end{align}
\textit{Step 2.} Taking the absolute value inside the integral and using $\abs{\phi_0}\le1$ removes the central node. What remains is a product of functions of one variable each, and since the density is a product, the integral factorizes,
\begin{align}
\abs I&\le\prod_{i=1}^n x_i, \nonumber\\
x_i&=\int\!\dd\lambda_i\,q_i(\lambda_i)\,\frac{\abs{\beta^{(i)}_0+\beta^{(i)}_1}}{2}. \label{eq:Ibound}
\end{align}
This is the only place where the independence of the sources is used. In the same way $\abs J\le\prod_iy_i$, where $y_i$ contains $\abs{\beta^{(i)}_0-\beta^{(i)}_1}/2$.

\textit{Step 3.} For real numbers $s,t\in[-1,1]$,
\begin{align}
\tfrac12\abs{s+t}+\tfrac12\abs{s-t}=\max(\abs s,\abs t)\le1. \label{eq:maxid}
\end{align}
Applied to $s=\beta^{(i)}_0$ and $t=\beta^{(i)}_1$ and integrated, it gives $x_i+y_i\le1$.

\textit{Step 4.} For nonnegative numbers the arithmetic mean is at least the geometric mean. Applied to the $n$ numbers $x_i/(x_i+y_i)$ and to the $n$ numbers $y_i/(x_i+y_i)$, and added, it gives
\begin{align}
\prod_i\left(\frac{x_i}{x_i+y_i}\right)^{\!1/n}+\prod_i\left(\frac{y_i}{x_i+y_i}\right)^{\!1/n}&\le\frac1n\sum_i\frac{x_i+y_i}{x_i+y_i}\nonumber\\
&=1. \label{eq:amgm}
\end{align}
Multiplying by $\prod_i(x_i+y_i)^{1/n}$ and using Step 3,
\begin{align}
\prod_ix_i^{1/n}+\prod_iy_i^{1/n}&\le\prod_i\left(x_i+y_i\right)^{1/n}\le1. \label{eq:lemmaxy}
\end{align}
Together with Step 2 this proves \cref{eq:ineq}. \qed

\textit{Proof of \cref{prop:open}.} The proof above used only one property of the end nodes, namely Step 3. For the open network we therefore need $\abs{u_0}+\abs{u_1}\le1$ for the mean values of the kernels \eqref{eq:kernel} at fixed $\lambda_i$. Given $\lambda_i$, Alice and Bob respond independently with means $\alpha_x,\beta_y\in[-1,1]$, so
\begin{align}
u_0&=\tfrac12\left(\alpha_0\beta_0-\alpha_1\beta_1\right), \nonumber\\
u_1&=\tfrac12\left(\alpha_0\beta_1+\alpha_1\beta_0\right). \label{eq:u0u1}
\end{align}
Each of $u_0$ and $u_1$ is linear in each of the four means. Hence $\abs{u_0}+\abs{u_1}$ is convex in each mean separately and is largest when every mean equals $\pm1$. For such values,
\begin{align}
(\alpha_0\beta_0)(\alpha_1\beta_1)=(\alpha_0\beta_1)(\alpha_1\beta_0). \label{eq:signs}
\end{align}
If this common sign is $+1$, then $\alpha_0\beta_0=\alpha_1\beta_1$ and $\alpha_0\beta_1=\alpha_1\beta_0$, so $u_0=0$ and $\abs{u_1}=1$. If it is $-1$, then $\abs{u_0}=1$ and $u_1=0$. In both cases $\abs{u_0}+\abs{u_1}=1$. Steps 2 and 4 then give \cref{eq:open} and its version with $n$th roots. Since $U_0$ and $U_1$ have operator norm at most one, the left-hand side of \cref{eq:open} cannot exceed $2$. \qed

\section{Proof of Proposition~\ref{prop:value}} \label{app:value}

\textit{Step 1: the end nodes.} Let Bob measure $B_y=\vect b_y\cdot\vect\sigma$. If the angle between the unit vectors $\vect b_0$ and $\vect b_1$ is $2\theta_1$, then $\abs{\vect b_0\pm\vect b_1}^2=2\pm2\cos2\theta_1$ and $(\vect b_0+\vect b_1)\cdot(\vect b_0-\vect b_1)=0$. Hence
\begin{align}
\vect b_0+\vect b_1&=2\cos\theta_1\,\vect n_1, \nonumber\\
\vect b_0-\vect b_1&=2\sin\theta_1\,\vect n_1', \label{eq:bn}
\end{align}
with $\theta_1\in[0,\pi/2]$ and orthonormal $\vect n_1,\vect n_1'$. The same holds for David with $\theta_2,\vect n_2,\vect n_2'$.

\textit{Step 2: factorization.} With $f_0=(A_\mu\otimes A_\mu)(C_\nu\otimes C_\nu)$, the correlator $I$ is the mean value of
\begin{align}
\left(A_\mu\otimes\tfrac{B_0+B_1}{2}\otimes C_\nu\right)_{\!1}\otimes\left(A_\mu\otimes\tfrac{D_0+D_1}{2}\otimes C_\nu\right)_{\!2}, \label{eq:Iop}
\end{align}
where the subscripts denote the particles of the two sources. For the state $\rho_1\otimes\rho_2$ its mean value is the product of the two factors. Using \cref{eq:bn} and the definition \eqref{eq:tau},
\begin{align}
I&=\cos\theta_1\cos\theta_2\,(\vect\tau_1\cdot\vect n_1)(\vect\tau_2\cdot\vect n_2), \nonumber\\
J&=\sin\theta_1\sin\theta_2\,(\vect\tau_1'\cdot\vect n_1')(\vect\tau_2'\cdot\vect n_2'), \label{eq:IJq}
\end{align}
where $\vect\tau_i$ and $\vect\tau_i'$ are the vectors $\vect\tau_{\mu\nu}$ and $\vect\tau_{\mu'\nu'}$ of source $i$.

\textit{Step 3: the angles.} Let $p$ and $q$ be the absolute values of the two products of scalar products in \cref{eq:IJq}. For nonnegative numbers $\sqrt{ab}+\sqrt{cd}\le\sqrt{a+c}\,\sqrt{b+d}$, which follows by squaring both sides. Hence
\begin{align}
\sqrt{\abs I}+\sqrt{\abs J}&\le\sqrt{\cos\theta_1\cos\theta_2+\sin\theta_1\sin\theta_2}\,\sqrt{p+q}\nonumber\\
&=\sqrt{\cos(\theta_1-\theta_2)}\,\sqrt{p+q}\nonumber\\
&\le\sqrt{p+q}, \label{eq:cs}
\end{align}
with equality for $\theta_1=\theta_2=\theta$ and $\tan^2\theta=q/p$.

\textit{Step 4: identical sources.} For $\rho_1=\rho_2=\rho$ the same inequality gives $p+q\le\sqrt{X_1}\sqrt{X_2}$, where $X_i=(\vect\tau\cdot\vect n_i)^2+(\vect\tau'\cdot\vect n_i')^2$. So it is enough to take $\vect n_1=\vect n_2=\vect n$ and $\vect n_1'=\vect n_2'=\vect n'$, and we must maximize
\begin{align}
X=(\vect\tau\cdot\vect n)^2+(\vect\tau'\cdot\vect n')^2,\qquad \vect n\perp\vect n'. \label{eq:Xmax}
\end{align}

\textit{Step 5: the best $\vect n'$.} For fixed $\vect n$, the vector $\vect n'$ lies in the plane perpendicular to $\vect n$. The overlap with $\vect\tau'$ is largest when $\vect n'$ points along the part of $\vect\tau'$ in that plane, whose squared length is $\abs{\vect\tau'}^2-(\vect\tau'\cdot\vect n)^2$. Therefore
\begin{align}
\max_{\vect n'\perp\vect n}X&=\abs{\vect\tau'}^2+(\vect\tau\cdot\vect n)^2-(\vect\tau'\cdot\vect n)^2\nonumber\\
&=\abs{\vect\tau'}^2+\vect n^{T}G\,\vect n, \label{eq:lmax}
\end{align}
with the symmetric matrix $G=\vect\tau\vect\tau^{T}-\vect\tau'\vect\tau'^{T}$.

\textit{Step 6: the best $\vect n$.} The maximum of $\vect n^TG\vect n$ over unit vectors is the largest eigenvalue $\lambda_{\max}$ of $G$. The matrix $G$ maps every vector into the plane spanned by $\vect\tau$ and $\vect\tau'$, so its third eigenvalue is zero. Within the plane,
\begin{align}
G\vect\tau&=\abs{\vect\tau}^2\vect\tau-g\,\vect\tau', \nonumber\\
G\vect\tau'&=g\,\vect\tau-\abs{\vect\tau'}^2\vect\tau', \label{eq:Gact}
\end{align}
with $g=\vect\tau\cdot\vect\tau'$. The two nonzero eigenvalues therefore have sum $\abs{\vect\tau}^2-\abs{\vect\tau'}^2$ and product $g^2-\abs{\vect\tau}^2\abs{\vect\tau'}^2$. Solving the quadratic equation gives
\begin{align}
\lambda_{\max}=\frac12\left[\abs{\vect\tau}^2-\abs{\vect\tau'}^2+\sqrt{s^2-4g^2}\right], \label{eq:lmax2}
\end{align}
with $s=\abs{\vect\tau}^2+\abs{\vect\tau'}^2$. Adding $\abs{\vect\tau'}^2$ gives $\Lambda$ of \cref{eq:Lambda}, and maximizing over the frames and indices gives \cref{eq:value}.

Finally, $\vect\tau_{\mu\nu}\cdot\vect n$ is the mean value of an observable with eigenvalues $\pm1$, namely $A_\mu\otimes\vect n\cdot\vect\sigma\otimes C_\nu$. Hence $\abs{\vect\tau_{\mu\nu}}\le1$. \qed

\section{Proof of Theorem~\ref{thm:svet}} \label{app:svet}

Let the three parties measure along $\vect a_x$, $\vect b_y$, $\vect c_z$, where $\vect b_y$ refers to the particle that goes to the end node in the network.

\textit{Step 1: sort by the end node.} Collecting the terms of \cref{eq:M,eq:Mp} that contain $B_0$ and those that contain $B_1$ gives
\begin{align}
M+M'&=(\mathsf P+\mathsf Q)B_0+(\mathsf P-\mathsf Q)B_1\nonumber\\
&=\mathsf P\,(B_0+B_1)+\mathsf Q\,(B_0-B_1), \label{eq:sort}
\end{align}
with $\mathsf P=A_0C_1+A_1C_0$ and $\mathsf Q=A_1C_1-A_0C_0$.

\textit{Step 2: decompose the settings.} We write the settings of Alice and Charlie as
\begin{align}
\vect a_{0,1}&=\cos\alpha\,\vect e\pm\sin\alpha\,\vect e', \nonumber\\
\vect c_{0,1}&=\cos\gamma\,\vect f\pm\sin\gamma\,\vect f', \label{eq:ac}
\end{align}
with $\vect e\perp\vect e'$ and $\vect f\perp\vect f'$. This is \cref{eq:bn} read backwards. For Bob we use \cref{eq:bn} with angle $\theta$ and directions $\vect n\perp\vect n'$.

\textit{Step 3: expand.} Inserting \cref{eq:ac} into $\mathsf P$ and $\mathsf Q$, the mixed terms cancel and we obtain
\begin{align}
\mathsf P&=2\cos\alpha\cos\gamma\ \vect e\otimes\vect f-2\sin\alpha\sin\gamma\ \vect e'\otimes\vect f', \nonumber\\
\mathsf Q&=-2\cos\alpha\sin\gamma\ \vect e\otimes\vect f'-2\sin\alpha\cos\gamma\ \vect e'\otimes\vect f, \label{eq:PQ}
\end{align}
where $\vect e\otimes\vect f$ stands for the observable $\vect e\cdot\vect\sigma\otimes\vect f\cdot\vect\sigma$ of Alice and Charlie.

\textit{Step 4: four correlators.} We name the correlators that appear,
\begin{align}
X&=T(\vect e,\vect n,\vect f), & Y&=T(\vect e',\vect n,\vect f'), \nonumber\\
Z&=T(\vect e,\vect n',\vect f'), & W&=T(\vect e',\vect n',\vect f). \label{eq:xyzw}
\end{align}
Taking the mean value of \cref{eq:sort} with \cref{eq:bn,eq:PQ} gives
\begin{align}
\frac{\avg{M+M'}}{4}=\ &\cos\theta\left(\cos\alpha\cos\gamma\,X-\sin\alpha\sin\gamma\,Y\right)\nonumber\\
&-\sin\theta\left(\cos\alpha\sin\gamma\,Z+\sin\alpha\cos\gamma\,W\right). \label{eq:svexp}
\end{align}

\textit{Step 5: bound.} The two coefficients in the first bracket satisfy $\abs{\cos\alpha\cos\gamma}+\abs{\sin\alpha\sin\gamma}\le1$, by the Cauchy-Schwarz inequality. So the first bracket is at most $\max(\abs X,\abs Y)$ in absolute value, and likewise the second bracket is at most $\max(\abs Z,\abs W)$. Using $\abs{\cos\theta\,u+\sin\theta\,w}\le\sqrt{u^2+w^2}$ and maximizing over all settings,
\begin{align}
\frac{\Sv(\rho)^2}{16}\le\max\left(X^2,Y^2\right)+\max\left(Z^2,W^2\right), \label{eq:svbound}
\end{align}
where $X,Y,Z,W$ are evaluated at the optimal settings.

\textit{Step 6: the network reaches the bound.} The right-hand side of \cref{eq:svbound} is one of $X^2+Z^2$, $X^2+W^2$, $Y^2+Z^2$ and $Y^2+W^2$. Let Alice use a frame that contains $\vect e$ and $\vect e'$, and Charlie a frame that contains $\vect f$ and $\vect f'$. Consider the case $X^2+Z^2$. For $I$ we use Alice's parity along $\vect e$ and Charlie's parity along $\vect f$, so that $\vect\tau\cdot\vect n=X$. For $J$ we use Alice's parity along $\vect e$ and Charlie's parity along $\vect f'$, so that $\vect\tau'\cdot\vect n'=Z$. By \cref{eq:Lambda}, $\Lambda(\vect\tau,\vect\tau')\ge X^2+Z^2$. The other three cases are identical, with the pairs of parities $(\vect e,\vect f)$ and $(\vect e',\vect f)$ for $X^2+W^2$, $(\vect e',\vect f')$ and $(\vect e,\vect f')$ for $Y^2+Z^2$, and $(\vect e',\vect f')$ and $(\vect e',\vect f)$ for $Y^2+W^2$. All these parities come from one Bell-state measurement per central node. Hence $\Nc(\rho)^2\ge\Sv(\rho)^2/16$. \qed

\section{Proof of Theorem~\ref{thm:mermin}} \label{app:mermin}

\textit{Step 1: network value.} Fix the settings of the end nodes. Let Charlie use the parities of two orthonormal directions $\vect c$ and $\vect c'$ of his frame for $g_0$ and $g_1$. As in Appendix~\ref{app:value}, two copies of $\rho$ give
\begin{align}
I_{\rm o}=(\vect u_0\cdot\vect c)^2,\qquad J_{\rm o}=(\vect u_1\cdot\vect c')^2, \label{eq:IJou}
\end{align}
so that
\begin{align}
\sqrt{\abs{I_{\rm o}}}+\sqrt{\abs{J_{\rm o}}}=\abs{\vect u_0\cdot\vect c}+\abs{\vect u_1\cdot\vect c'}. \label{eq:openval}
\end{align}

\textit{Step 2: best frame.} Only the components of $\vect c$ and $\vect c'$ in the plane of $\vect u_0$ and $\vect u_1$ matter, and it is best to take both vectors in this plane. Let $\vect c$ make the angle $\varphi$ with a fixed axis in the plane and $\vect c'=\pm\,\vect c_\perp$, where $\vect c_\perp$ is $\vect c$ rotated by $\pi/2$. Removing the absolute values in \cref{eq:openval} by a choice of signs, we must maximize
\begin{align}
\vect u_0\cdot\vect c\pm\vect u_1\cdot\vect c_\perp=a_\pm\cos\varphi+b_\pm\sin\varphi \label{eq:abphi}
\end{align}
over $\varphi$, where $a_\pm,b_\pm$ are fixed numbers. The maximum is $\sqrt{a_\pm^2+b_\pm^2}$, and a short calculation in components gives
\begin{align}
a_\pm^2+b_\pm^2=\abs{\vect u_0}^2+\abs{\vect u_1}^2\pm2\,\hat{\vect k}\cdot(\vect u_0\times\vect u_1), \label{eq:abpm}
\end{align}
with $\hat{\vect k}$ the unit normal of the plane. The better sign turns the last term into $2\abs{\vect u_0\times\vect u_1}$, which proves \cref{eq:No}. (Equivalently, the maximum is the sum of the two singular values of the $2\times3$ matrix with rows $\vect u_0$ and $\vect u_1$.)

\textit{Step 3: Mermin value.} Since $M=2(U_0C_1+U_1C_0)$, one copy gives $\avg M=2(\vect u_0\cdot\vect c_1+\vect u_1\cdot\vect c_0)$. The third party may choose $\vect c_1$ along $\vect u_0$ and $\vect c_0$ along $\vect u_1$, so that
\begin{align}
\frac{\Me(\rho)}{2}=\max\left(\abs{\vect u_0}+\abs{\vect u_1}\right). \label{eq:Mval}
\end{align}

\textit{Step 4: compare.} Since $\abs{\vect u_0\times\vect u_1}\le\abs{\vect u_0}\abs{\vect u_1}$, the bracket in \cref{eq:No} is at most $(\abs{\vect u_0}+\abs{\vect u_1})^2$ for the same settings. This proves \cref{eq:mermin}. If $\vect u_0\perp\vect u_1$ at the settings that maximize \cref{eq:Mval}, the two brackets coincide there and equality holds. For the generalized GHZ state and settings $\sigma_1,\sigma_2$ one finds $\vect u_0=\sin2\alpha\,\hat{\vect x}$ and $\vect u_1=-\sin2\alpha\,\hat{\vect y}$. \qed

\section{Proof of Theorem~\ref{thm:parity}} \label{app:parity}

We use two elementary facts.

\textit{Fact 1.} If $P$ and $Q$ are Hermitian, anticommute and have operator norm at most one, then every state satisfies
\begin{align}
\avg P^2+\avg Q^2\le1. \label{eq:anti}
\end{align}
To see this, take real numbers $x,y$ with $x^2+y^2=1$. Because $PQ+QP=0$, the cross terms cancel in
\begin{align}
(xP+yQ)^2=x^2P^2+y^2Q^2\le\id. \label{eq:antisq}
\end{align}
So $xP+yQ$ has norm at most one and $\abs{x\avg P+y\avg Q}\le1$. Choosing $(x,y)$ parallel to $(\avg P,\avg Q)$ gives \cref{eq:anti}.

\textit{Fact 2.} If $p,q\ge0$, $p^2+q^2\le1$ and $0<r\le1$, then
\begin{align}
p^{r}+q^{r}\le2\left(\frac{p^2+q^2}{2}\right)^{\!r/2}\le2^{1-r/2}. \label{eq:conc}
\end{align}
The first inequality holds because $t\mapsto t^{r/2}$ is concave, so the mean of its values at $p^2$ and $q^2$ is at most its value at the mean.

\textit{Step 1: removing the sign functions.} Collect all outcomes of the central nodes in one label $k$, with probability $\pi_k$. Let $\avg{\,\cdot\,}_k$ be the mean value in the state of the end nodes conditioned on $k$. For any operator $P$ of the end nodes and any sign function $f$ of $k$,
\begin{align}
\abs{\avg{fP}}=\Big|\sum_k\pi_kf(k)\avg P_k\Big|\le\sum_k\pi_k\abs{\avg P_k}\equiv\bar p. \label{eq:pbar}
\end{align}
We define $\bar q$ in the same way for an operator $Q$. If $\norm{P},\norm{Q}\le1$, then $\bar p,\bar q\le1$. If in addition $P$ and $Q$ anticommute, Fact 1 holds for every $k$. Since the square of a mean is at most the mean of the squares,
\begin{align}
\bar p^{\,2}+\bar q^{\,2}&\le\sum_k\pi_k\left(\avg P_k^2+\avg Q_k^2\right)\le1. \label{eq:pq}
\end{align}

\textit{Step 2: closed network.} For qubit end nodes, \cref{eq:bn} gives
\begin{align}
I&=\prod_i\cos\theta_i\,\avg{f_0P}, \nonumber\\
J&=\prod_i\sin\theta_i\,\avg{f_1Q}, \label{eq:IJPQ}
\end{align}
where $P$ is the tensor product of the operators $\vect n_i\cdot\vect\sigma$ and $Q$ that of $\vect n_i'\cdot\vect\sigma$. By \cref{eq:pbar}, $\abs I\le\prod_i\cos\theta_i\,\bar p$ and $\abs J\le\prod_i\sin\theta_i\,\bar q$. We apply the first inequality of \cref{eq:lemmaxy} with $x_i=\cos\theta_i\,\bar p^{1/n}$ and $y_i=\sin\theta_i\,\bar q^{1/n}$. Each factor then satisfies $\cos\theta_i\,\bar p^{1/n}+\sin\theta_i\,\bar q^{1/n}\le(\bar p^{2/n}+\bar q^{2/n})^{1/2}$, so
\begin{align}
\abs I^{1/n}+\abs J^{1/n}&\le\sqrt{\bar p^{\,2/n}+\bar q^{\,2/n}}. \label{eq:closedbound}
\end{align}
Since $\bar p,\bar q\le1$, this is at most $\sqrt2$. For odd $n$, each of the $n$ factors of $P$ anticommutes with the corresponding factor of $Q$, so $PQ=-QP$. Then \cref{eq:pq} and Fact 2 with $r=2/n$ give $\bar p^{\,2/n}+\bar q^{\,2/n}\le2^{1-1/n}$, and \cref{eq:closedbound} becomes $2^{1/2-1/(2n)}$.

\textit{Step 3: open network.} Now $P$ and $Q$ are the tensor products of the kernels $U_0$ and $U_1$ of the $n$ sources. Using only $A_x^2=B_y^2=\id$ and the fact that Alice's and Bob's operators commute,
\begin{align}
4\,U_0U_1&=B_0B_1+A_0A_1-A_1A_0-B_1B_0, \nonumber\\
4\,U_1U_0&=B_1B_0-A_0A_1+A_1A_0-B_0B_1, \label{eq:antiU}
\end{align}
so $U_0U_1+U_1U_0=0$ in any dimension. The kernels have norm at most one, which gives the bound $2$. For odd $n$, $P$ and $Q$ anticommute, and \cref{eq:pbar,eq:pq} together with Fact 2 for $r=1/n$ give
\begin{align}
\abs{I_{\rm o}}^{1/n}+\abs{J_{\rm o}}^{1/n}\le\bar p^{1/n}+\bar q^{1/n}\le2^{1-\frac1{2n}}. \label{eq:openbound}
\end{align}

\textit{Step 4: tightness.} Let all sources emit GHZ states. For even $n$, group the sources in pairs. In the closed network both central nodes perform Bell-state measurements on every pair. We take $f_0$ as the product of Alice's $\sigma_3\otimes\sigma_3$ parities and $f_1$ as the product of the $\sigma_1\otimes\sigma_1$ parities of both central nodes, and all end nodes measure $(\sigma_3\pm\sigma_1)/\sqrt2$. Then $\abs I=\abs J=2^{-n/2}$ and the value is $\sqrt2$. In the open network Charlie performs Bell-state measurements on the pairs, $g_0$ and $g_1$ are the products of the $\sigma_1\otimes\sigma_1$ and $\sigma_2\otimes\sigma_2$ parities, and $I_{\rm o}=J_{\rm o}=1$.

For odd $n$ one source is left over. In the closed network both central nodes measure $\sigma_1$ on its particles, and both $f_0$ and $f_1$ are multiplied by the two results. Its end node measures $\sigma_1$ and $\sigma_2$, so that $\theta=\pi/4$, $\vect n=(\hat{\vect x}+\hat{\vect y})/\sqrt2$ and $\vect n'=(\hat{\vect x}-\hat{\vect y})/\sqrt2$. With $\vect\tau_{11}=\hat{\vect x}$ this source contributes $\cos\theta\,\abs{\vect\tau_{11}\cdot\vect n}=1/2$ to $\abs I$ and $\sin\theta\,\abs{\vect\tau_{11}\cdot\vect n'}=1/2$ to $\abs J$. For the paired sources the end nodes measure $(\sigma_1\pm\sigma_2)/\sqrt2$, $f_0$ uses the $\sigma_1\otimes\sigma_1$ parities of both central nodes, and $f_1$ uses Alice's $\sigma_2\otimes\sigma_2$ parity and Charlie's $\sigma_1\otimes\sigma_1$ parity. Each paired source contributes $1/\sqrt2$. In total
\begin{align}
\abs I=\abs J=2^{-(n-1)/2}\cdot\tfrac12=2^{-(n+1)/2}, \label{eq:oddIJ}
\end{align}
and the value is $2\cdot2^{-(n+1)/(2n)}=2^{1/2-1/(2n)}$. In the open network Charlie measures the left-over qubit along $(\hat{\vect x}+\hat{\vect y})/\sqrt2$ and multiplies $g_0$ and $g_1$ by the result. This gives $\abs{I_{\rm o}}=\abs{J_{\rm o}}=1/\sqrt2$ and the value $2^{1-1/(2n)}$. \qed

\section{Entanglement swapping} \label{app:swap}

\textit{Proof of \cref{prop:swap}.} We first give the short argument for two sources. The GHZ state is stabilized by $\sigma_3\otimes\sigma_3\otimes\id$, $\id\otimes\sigma_3\otimes\sigma_3$ and $\sigma_1\otimes\sigma_1\otimes\sigma_1$, with the particle of the end node in the middle. Multiplying the stabilizers of the two sources gives, with certainty,
\begin{align}
\sigma_3^B\sigma_3^D&=\sigma_3^{A_1}\sigma_3^{A_2}=a_z, \nonumber\\
\sigma_3^B\sigma_3^D&=\sigma_3^{C_1}\sigma_3^{C_2}=c_z, \nonumber\\
\sigma_1^B\sigma_1^D&=\big(\sigma_1^{A_1}\sigma_1^{A_2}\big)\big(\sigma_1^{C_1}\sigma_1^{C_2}\big)=a_xc_x. \label{eq:stabswap}
\end{align}
The first two lines give $c_z=a_z$, and together with the third they fix one Bell state of Bob and David, \cref{eq:swap}.

For $n$ sources it is easier to work with the state itself,
\begin{align}
\ket{\GHZ}^{\otimes n}=2^{-n/2}\sum_{\vect s\in\{0,1\}^n}\ket{\vect s}_A\ket{\vect s}_B\ket{\vect s}_C, \label{eq:ghzn}
\end{align}
where $A$ and $C$ collect the qubits of the central nodes and $B$ those of the end nodes. The GHZ basis of $n$ qubits consists of the states $(\ket{\vect k}\pm\ket{\bar{\vect k}})/\sqrt2$, where $\bar{\vect k}$ is the complement of the bit string $\vect k$. If Alice finds $\vect k$ with sign $\epsilon_A$, only the terms $\vect s=\vect k$ and $\vect s=\bar{\vect k}$ survive, and the remaining state is $(\ket{\vect k}_B\ket{\vect k}_C+\epsilon_A\ket{\bar{\vect k}}_B\ket{\bar{\vect k}}_C)/\sqrt2$. Charlie can then only find the same pair $\{\vect k,\bar{\vect k}\}$. With his sign $\epsilon_C$, the end nodes are left in
\begin{align}
\frac{1}{\sqrt2}\left(\ket{\vect k}+\epsilon_A\epsilon_C\ket{\bar{\vect k}}\right). \label{eq:swapped}
\end{align}
Without Charlie's outcome, $\epsilon_C$ is unknown and equally likely to be $\pm1$. The average of \cref{eq:swapped} over $\epsilon_C$ is the separable mixture of $\proj{\vect k}$ and $\proj{\bar{\vect k}}$, because the cross terms cancel. The same holds with the roles of the central nodes exchanged. \qed

\textit{Noisy GHZ sources.} The Bell projectors can be written as
\begin{align}
\Phi^\pm&=\tfrac14\left(\id\pm\sigma_1\sigma_1\mp\sigma_2\sigma_2+\sigma_3\sigma_3\right), \nonumber\\
\Psi^\pm&=\tfrac14\left(\id\pm\sigma_1\sigma_1\pm\sigma_2\sigma_2-\sigma_3\sigma_3\right). \label{eq:bellproj}
\end{align}
Suppose both central nodes find $\Phi^+$. Expanding both projectors gives terms with Pauli index $\mu$ on Alice's two qubits and $\nu$ on Charlie's two qubits. Each term factorizes into one factor per source. The probability of the outcome is
\begin{align}
\pi=\tfrac1{16}\left(1+T_{303}^2\right)=\tfrac1{16}\left(1+v^2\right), \label{eq:piswap}
\end{align}
because the only nonzero entries of the form $T_{\mu0\nu}$ are $T_{000}=1$ and $T_{303}=v$. The correlators of Bob and David conditioned on this outcome are
\begin{align}
\avg{\sigma_3\sigma_3}&=\frac{T_{330}^2+T_{033}^2}{16\pi}=\frac{2v^2}{1+v^2}, \nonumber\\
\avg{\sigma_1\sigma_1}&=\frac{T_{111}^2+T_{212}^2}{16\pi}=\frac{2v^2}{1+v^2}, \nonumber\\
\avg{\sigma_2\sigma_2}&=-\frac{T_{122}^2+T_{221}^2}{16\pi}=-\frac{2v^2}{1+v^2}. \label{eq:corrswap}
\end{align}
The minus sign in the last line comes from the sign of $\sigma_2\sigma_2$ in $\Phi^+$. All other correlators and the local Bloch vectors vanish. This is the Werner state of \cref{eq:w}. The eight outcomes with $c_z=a_z$ are related to this one by local Pauli operations and have total probability $(1+v^2)/2$. For the outcomes with $c_z\neq a_z$, the projector at one central node is a $\Psi$ projector, the two terms in each numerator of \cref{eq:corrswap} enter with opposite signs and cancel, and Bob and David are left maximally mixed. For three sources, a direct calculation gives the heralding probability $(1+3v^2)/4$ and a state diagonal in the GHZ basis with coherence $\bra{000}\rho\ket{111}=2v^3/(1+3v^2)$. Its Mermin and Svetlichny values, evaluated numerically, are $8$ and $8\sqrt2$ times this coherence, which gives the expressions quoted in \cref{sec:nlocal}.

\section{Numerical methods} \label{app:numerics}

The proofs above do not rely on numerics. The computations described here serve two purposes. They are consistency checks of the analytical results, and they provide the values marked as numerical in \cref{tab:values}. The code is available from the authors upon reasonable request.

\textit{Parametrization.} Unit vectors are parametrized by two spherical angles. The frames of the central nodes are rotations parametrized by a rotation vector in $\mathbb R^3$, which covers all rotations. Redundant parametrizations, for example rotations that only permute the axes, are harmless because \cref{eq:value} already maximizes over all assignments of parities. General projective measurements of a central node on two qubits are given by a unitary $U=\exp(iH)$ whose columns form the measurement basis, with $H$ a Hermitian matrix of $16$ real parameters. For such measurements the best sign functions are chosen outcome by outcome, as in \cref{eq:pbar}.

\textit{Optimization.} $\Nc$ is maximized with the Nelder-Mead method over the six parameters of the two frames, using $40$ random starting points and tolerances of $10^{-7}$ in the parameters and $10^{-10}$ in the value. $\Sv$, $\Me$ and $\No$ are maximized with the BFGS method over the end-node angles, using $60$, $60$ and $40$ random starting points. For $\No$ the frame of the central node is optimized analytically by \cref{eq:No}. $\Nc^{\rm prod}$ is maximized with the L-BFGS-B method over $32$ unitary parameters and $10$ end-node angles. All reported values are stable to the quoted digits when the number of starting points is doubled. For the $W$ state, $\Nc^{\rm prod}$ changed by less than $10^{-6}$ between $16$ and $40$ starting points. Because local optimization finds lower bounds on a maximum, all numerical maxima are lower bounds.

\textit{Critical visibilities.} For a fixed strategy, every network value is proportional to $v$. To see this, expand $\rho_v^{\otimes2}$ into the four terms $v^2\psi\otimes\psi$, $v(1-v)\,\psi\otimes\id/8$, $(1-v)v\,\id/8\otimes\psi$ and $(1-v)^2\id/8\otimes\id/8$. In every term that contains $\id/8$, the particle of one end node is maximally mixed and uncorrelated with all other particles. Its traceless observables then have mean value zero, and the term does not contribute to $I$ or $J$. Hence $I$ and $J$ are multiplied by $v^2$, and $\sqrt{\abs I}+\sqrt{\abs J}$ by $v$. The same holds for the maximum over any class of strategies. A pure-state value $K$ therefore gives the unique threshold $v_c=1/K$. For the single-copy values $\Me$ and $\Sv$ the same scaling follows from $T_{ikm}\to vT_{ikm}$.

\textit{Random states.} The states of \cref{fig:random} were generated as follows. Pure states are normalized vectors with independent complex Gaussian entries, which are distributed uniformly with respect to the Haar measure. Rank-two states are $GG^\dagger/\Tr[GG^\dagger]$, with $G$ an $8\times2$ matrix of independent complex Gaussian entries. One third of all states were then mixed in equal proportion with the GHZ or the $W$ state, in order to include states with larger violations. The ensemble is therefore not uniform with respect to any standard measure. The smallest difference $\Nc-\Sv/4$ found was $0.03$, and the largest difference $\No-\Me/2$ vanished within numerical precision.

\textit{Further checks.} The local bounds \eqref{eq:ineq} and \eqref{eq:open} were tested on $2\times10^4$ random hidden-variable models with hidden variables taking three values, for $n=2,3,4$, and no violation was found. \Cref{eq:value} was compared with a direct evaluation on the six-qubit state $\rho\otimes\rho$ for random states and frames, and the two agreed to machine precision. The bounds of \cref{thm:parity} were tested for a star with $n\le4$ Bell pairs and a central node that measures in an arbitrary basis of its $n$ qubits. The optimization returned $1.0000$, $1.4142$, $1.2599$ and $1.4110$ for $n=1,\dots,4$, in agreement with the theorem. For $n=4$ the search over $256$ parameters stopped slightly below $\sqrt2$. For the open network with three GHZ sources it returned $1.7818$, which equals $2^{5/6}$. For the biseparable state of \cref{eq:bisep} we found $\Nc=\sqrt2$, $\Sv/4=1$, $\Me/2=\sqrt2$ and $\No=1$, and $\No=\sqrt2$ when the unentangled particle is moved to an end node.

\bibliography{refs2}
\end{document}